\documentclass[letterpaper, 11pt]{article}

\newenvironment{proof}{{\bf Proof:}}{\hspace*{\fill}\(\Box\)}

\title{Iterative Byzantine Vector Consensus in Incomplete Graphs~\thanks{This research is supported
in part by
 National
Science Foundation award CNS-1059540
and Army Research Office grant W-911-NF-0710287. Any opinions, findings, and conclusions or recommendations expressed here are those of the authors and do not
necessarily reflect the views of the funding agencies or the U.S. government.}
}

\author{{\normalsize\bf Nitin H. Vaidya}\\
 \normalsize Department of Electrical and Computer Engineering\\
 \normalsize University of Illinois at Urbana-Champaign\\
\normalsize nhv@illinois.edu\\}

\date{July 8, 2013}

\usepackage{graphicx}
\newcommand{\comment}[1]{}

\usepackage{graphicx}
\usepackage{latexsym}

\newcommand{\nchoosek}[2]{{#1 \choose #2}}

\newcommand{\algo}{{\tt Byz-Iter}}
\newcommand{\matrixm}{{\bf M}}

\newcommand{\matrixh}{{\bf H}}
\newcommand{\graphh}{H}
\newcommand{\terminate}{t_{end}}

\newcommand{\Xeightarrow}[1]{\stackrel{~#1}{\longrightarrow}}
\newcommand{\fArrow}{\Xeightarrow{f}}
\newcommand{\dfArrow}{\Xeightarrow{df}}
\newcommand{\cArrow}{\Xeightarrow{c}}

\newcommand{\scripte}{\mathcal{E}}
\newcommand{\scriptf}{\mathcal{F}}
\newcommand{\scripth}{\mathcal{H}}

\newcommand{\scriptv}{\mathcal{V}}

\newtheorem{theorem}{Theorem}

\newtheorem{claim}{Claim}

\newtheorem{definition}{Definition}

\newtheorem{lemma}{Lemma}

\def\noflash#1{\setbox0=\hbox{#1}\hbox to 1\wd0{\hfill}}

\newcommand{\bfA}{{\bf A}}

\newcommand{\bfB}{{\bf B}}

\newcommand{\bfH}{{\bf H}}

\newcommand{\bfM}{{\bf M}}

\newcommand{\bfQ}{{\bf Q}}

\newcommand{\bfx}{{\bf x}}
\newcommand{\bfv}{{\bf v}}

\newcommand{\bfr}{{\bf r}}

\newcommand{\bfz}{{\bf z}}

\newcommand{\HH}{\mathcal{H}}

\begin{document}

\maketitle


\begin{abstract}
%
%
%
%
%
This work addresses {\em Byzantine vector consensus} (BVC), wherein the input at each process is a $d$-dimensional vector of reals, and each process is expected to decide on a {\em decision vector} that is in the {\em convex hull} of the input vectors at the fault-free processes  \cite{mendes13stoc,vaidya13podc}.  The input {\em vector} at each process may also be viewed as a {\em point} in the $d$-dimensional Euclidean space ${\bf R}^d$, where $d>0$ is a finite integer. 
Recent work \cite{mendes13stoc,vaidya13podc} has addressed Byzantine vector consensus in systems that can be modeled
by a {\em complete} graph.
This paper considers Byzantine vector consensus in {\em incomplete} graphs. In particular, we address a particular class of
{\em iterative} algorithms in incomplete graphs, and prove a necessary condition, and a sufficient condition,
for the graphs to be able to solve the vector consensus problem iteratively.
We present an iterative Byzantine vector consensus algorithm, and prove it correct
under the sufficient condition.
The necessary
condition presented in this paper for vector consensus does not match with the sufficient condition for $d>1$; thus,
a weaker condition may potentially suffice for Byzantine vector consensus.
\end{abstract}



\setcounter{page}{1}

\newcommand{\bfe}{{\bf e}}
\newcommand{\zero}{{\bf zero}}

\newcommand{\sg}{{\mathcal G}}
\newcommand{\sv}{{\mathcal V}}
\newcommand{\sF}{{\mathcal F}}
\newcommand{\se}{{\mathcal E}}
\newcommand{\sw}{{\mathcal W}}
\newcommand{\ms}{{\mathcal S}}

\section{Introduction}

This work addresses {\em Byzantine vector consensus} (BVC), wherein the input at each process is a $d$-dimensional vector of reals, and each process is expected to decide on a {\em decision vector} that is in the {\em convex hull} of the input vectors at the fault-free processes  \cite{mendes13stoc,vaidya13podc}.  The input {\em vector} at each process may also be viewed as a {\em point} in the $d$-dimensional Euclidean space ${\bf R}^d$, where $d>0$ is a finite integer.  Due to this correspondence, we use the terms {\em point} and {\em vector} interchangeably.
Recent work \cite{mendes13stoc,vaidya13podc} has addressed Byzantine vector consensus in systems that can be modeled
by a {\em complete} graph.
The correctness conditions for Byzantine vector consensus (elaborated below) cannot be satisfied by
independently performing consensus on each element of the input vectors; therefore, new algorithms are necessary.
Here we consider Byzantine vector consensus in {\em incomplete} graphs. In particular, we address a particular class of
{\em iterative} algorithms in incomplete graphs, and prove a necessary condition, and a sufficient condition,
for the graphs to be able to solve the vector consensus problem iteratively.
The paper extends our past work on {\em scalar} consensus in incomplete graphs in presence of Byzantine faults \cite{vaidya12podc},
which yielded an
exact characterization of graphs in which the problem is solvable.
We present an iterative Byzantine vector consensus algorithm, and prove it correct
under the sufficient condition;
the proof follows a structure
previously used in our work to prove correctness of other consensus algorithms
\cite{vaidya_matrix_IABC,tseng13icdcn}.

The necessary
condition presented in this paper for vector consensus does not match with the sufficient condition for $d>1$; thus, it is possible that
a weaker condition may also suffice for Byzantine vector consensus. We hope that this paper will motivate further work on identifying
the tight sufficient condition.

In other related work \cite{tseng13optconvex}, we present another generalization of the consensus problem considered in \cite{mendes13stoc,vaidya13podc}. 
In particular, \cite{tseng13optconvex} considers the problem of deciding on a {\em convex hull} (instead of just one point)
that is contained in the convex hull of the inputs at the fault-free nodes. 

The paper is organized as follows. Section \ref{s_model} presents our system model.
The iterative algorithm structure considered in our work is presented in Section \ref{s_structure}.
Section \ref{s_nec} presents a necessary condition, and Section \ref{s_suff} presents a sufficient condition.
Section \ref{s_suff} also
presents an iterative algorithm and proves
its correctness under the sufficient condition.
The paper concludes with a summary in Section \ref{s_summary}.

\section{System Model}
\label{s_model}
 
The system is assumed to be {\em synchronous}.\footnote{Analogous results can be similarly derived for asynchronous systems, using the
asynchronous algorithm structure presented
in \cite{vaidya12podc} for the case of $d=1$.}
The communication network is modeled as a simple {\em directed} graph $G(\scriptv,\scripte)$, where $\scriptv=\{1,\dots,n\}$ is the set of $n$ processes, and $\scripte$ is the set of directed edges between the processes in $\scriptv$. Thus, $|\sv|=n$. We assume that $n\geq 2$, since the consensus problem for $n=1$ is trivial.
 Process $i$ can reliably transmit messages to process $j$, $j\neq i$, if and only if
the directed edge $(i,j)$ is in $\scripte$.
Each process can send messages to itself as well, however,
for convenience of presentation, we \underline{exclude} self-loops from set $\scripte$.
That is, $(i,i)\not\in\scripte$ for $i\in\scriptv$.
We will use the terms {\em edge}
and {\em link} interchangeably.

For each process $i$, let $N_i^-$ be the set of processes from which $i$ has incoming
edges.
That is, $N_i^- = \{\, j ~|~ (j,i)\in \scripte\, \}$.
Similarly, define $N_i^+$ as the set of processes to which process $i$
has outgoing edges. That is, $N_i^+ = \{\, j ~|~ (i,j)\in \scripte\, \}$.
Since we exclude self-loops from $\scripte$,
$i\not\in N_i^-$ and $i\not\in N_i^+$. 
However, we note again that each process can indeed send messages to itself.

We consider the Byzantine failure model, with up to $f$ processes becoming faulty. A faulty process may {\em misbehave} arbitrarily. The faulty processes may potentially collaborate with each other. Moreover, the faulty processes are assumed to have a complete knowledge of the execution of
the algorithm, including the states of all the processes,
contents of messages the other processes send to each other,
the algorithm specification, and the network topology.

\paragraph{Notation:} We use the notation $|X|$ to denote the size of a
set or a multiset, and the notation $\|x\|$ to denote the absolute value
of a real number $x$.

\section{Byzantine Vector Consensus and Iterative Algorithms}
\label{s_structure}

\paragraph{Byzantine vector consensus:}

We are interested in {\bf iterative} algorithms that satisfy the
following conditions in presence of up to $f$ Byzantine faulty processes:
\begin{itemize}

\item {\em Termination}: Each fault-free process must terminate after a finite number of iterations.

\item {\em Validity}: The state of each fault-free process at the end of
\underline{\bf each iteration} must be in the convex hull of the $d$-dimensional input vectors at the fault-free processes.

\item {\em $\epsilon$-Agreement}: When the algorithm terminates,
the $l$-th elements of the decision vectors at any two fault-free processes,
where $1\leq l\leq d$,
must be within $\epsilon$ of each other, where $\epsilon>0$ is a
pre-defined constant.

\end{itemize}
Any information carried over by a process from iteration $t$ to iteration
$t+1$ is considered the state of process $t$ at the end of iteration $t$. 
The above {\em validity} condition forces the algorithms to maintain
``minimal'' 
state, for instance, precluding the possibility of remembering messages received in several of the past iterations, or remembering the history of detected misbehavior of the neighbors.
Therefore,
we focus on algorithms with a simple iterative structure, described below. 
\paragraph{Iterative structure:}
Each process $i$ maintains a state variable $\bfv_i$, which is a $d$-dimensional vector.
The initial state of process $i$ is denoted as
$\bfv_i[0]$, and it equals the {\em input}\, provided to process $i$.
For $t\geq 1$, $\bfv_i[t]$ denotes the state
of process $i$ at the {\em end}\, of the $t$-th iteration of the algorithm.
At the {\em start} of the $t$-th iteration ($t\geq 1$), the state of
process $i$ is $\bfv_i[t-1]$.
The iterative algorithms of interest will require each process $i$
to perform the following three steps in the $t$-th iteration.
Each ``value'' referred in the algorithm below is a
$d$-dimensional vector (or, equivalently, a point
in the $d$-dimensional Euclidean space).
\begin{enumerate}
\item {\em Transmit step:} Transmit current state, namely $\bfv_i[t-1]$, on all outgoing edges
 to processes in $N_i^+$.

\item {\em Receive step:} Receive values on all incoming edges from processes in $N_i^-$. 
Denote by $r_i[t]$ the multiset\footnote{The same value may occur
multiple times in a multiset.} of values received by process $i$ from its
neighbors. The size of multiset $r_i[t]$ is $|N_i^-|$.

\item {\em Update step:} Process $i$ updates its state using a transition function $T_i$ as
follows. $T_i$ is a part of the specification of the algorithm, and takes
as input the multiset $r_i[t]$ and state $\bfv_i[t-1]$.
\begin{eqnarray}
\bfv_i[t] & = &  T_i ~( ~r_i[t]\,,\,\bfv_i[t-1] ~)
\label{eq:Z_i}
\end{eqnarray}


\end{enumerate}
The decision (or output) of each process equals its state when the algorithm terminates. \\

We assume that each element of the input vector at each fault-free
process is lower bounded by a constant $\mu$ and upper bounded by a constant $U$.
The iterative algorithm may terminate after a number of rounds that is a function of $\mu$ and $U$.
$\mu$ and $U$ are assumed to be known {\em a priori}.
This assumption holds in many practical systems, because the
input vector elements represent quantities that are constrained. For instance,
if the input vectors are probability vectors, then $U=1$ and $\mu=0$.
If the input vectors represent locations in 3-dimensional space occupied
by mobile robots, then $U$ and $\mu$ are determined by the boundary of
the region in which the robots are allowed to operate.

In Section \ref{s_nec}, we develop a necessary condition that the graph $G(\sv,\se)$
must satisfy in order for the Byzantine vector consensus algorithm to be solvable
using the above iterative structure.
In Section \ref{s_suff}, we develope a sufficient condition, such that 
the Byzantine vector consensus algorithm is solvable
using the above iterative structure in any graph that satisfies this condition.
We present an iterative algorithm, and prove its correctness under the
sufficient condition.

\section{A Necessary Condition}
\label{s_nec}

Hereafter, when we refer to an iterative algorithm, we mean an algorithm with
the iterative structure specified in the previous section.
In this section, we state a necessary condition on graph $G(\sv,\se)$ to be able to
achieve Byzantine vector consensus using an iterative algorithm.
First we introduce some notations.

\begin{definition}
~
\begin{itemize}
\item Define $\bfe_0$ to be a $d$-dimensional vector with all its elements equal to 0.
	Thus, $\bfe_0$ corresponds to the origin in the $d$-dimensional Euclidean space.
\item Define $\bfe_i$, $1\leq i\leq d$, to be a $d$-dimensional vector with the $i$-th element equal to $2\epsilon$,
and the remaining elements equal to 0.
Recall that $\epsilon$ is the parameter of the $\epsilon$-agreement condition.
\end{itemize}
\end{definition}

\begin{definition}
\label{def:absorb}
For non-empty disjoint sets of processes $A$ and $B$, and a non-negative integer $c$,
\begin{itemize}
\item $A \cArrow B$ if and only if there exists a process $v\in B$ that has at least $c+1$ incoming
edges from processes in $A$, i.e., $|N_v^-\cap A|\geq c+1$.
\item $A\not\cArrow B$ iff $A\cArrow B$ is {\em not} true.
\end{itemize}
\end{definition}

~

\begin{definition}
\label{e_convexhull}
$\HH(X)$ denotes the convex hull of a multiset of points $X$.
\end{definition}

\noindent
Now we state the necessary condition.

\paragraph{Condition NC:}~ 
{\em
For any partition
$V_0,V_1,\cdots,V_p$, $C$, $F$ of set $\sv$, where
$1\leq p\leq d$, $V_k\neq\emptyset$ for $0\leq k\leq p$, and $|F|\leq f$,
there exist $i,j$  ($0\leq i,j\leq p$, $i\neq j$), such that
\[ V_i\cup C \fArrow V_j \]
That is, there are $f+1$ incoming links from
processes in $V_i\cup C$ to some process in $V_j$. \\
}

\begin{lemma}
\label{l_nc1}
If the Byzantine vector consensus problem can be solved using an iterative algorithm in $G(\sv,\se)$,
then $G(\sv,\se)$ satisfies Condition NC.
\end{lemma}
\begin{proof}
The proof is by contradiction.
Suppose that Condition NC is not true. Then there exists
a certain partition $V_0,V_1,\cdots,V_p,C,F$ such that
$V_k\neq\emptyset$ ($1\leq k\leq p$), $|F|\leq f$, and
for $0\leq i,k\leq p$, $V_k\cup C\not\fArrow V_i$.

Let the initial state of each process
in $V_i$ be $\bfe_i$ ($0\leq i\leq p$).
Suppose that all the processes in set $F$ are faulty.
For each link $(j,k)$ such that $j\in F$ and $k\in V_i$ ($0\leq i\leq p$),
the faulty process $j$ sends value $\bfe_i$ to process $j$ 
in each iteration.

We now prove by induction that if the iterative algorithm satisfies the validity
condition then the state of each fault-free process 
$j\in V_i$ at the start of iteration $t$ equals $\bfe_i$, for all $t>0$.
The claim is true for $t=1$ by assumption on the inputs at the fault-free processes.
Now suppose that the claim is true through iteration $t$, and prove it for iteration $t+1$.
Thus, the state of each fault-free process in $V_i$ at the start of iteration $t$ equals
$\bfe_i$, $0\leq i\leq p$.

Consider any fault-free process $j\in V_i$, where $0\leq i\leq p$.
In iteration $t$,
process $j$ will receive $\bfv_g[t-1]$ from each
fault-free incoming neighbor $g$,
and receive $\bfe_i$ from each faulty incoming neighbor. These received values
form the multiset $r_j[t]$.
Since the condition in the lemma is assumed to be false, for 
any $k\neq i$, $0\leq k\leq p$, we have
\[
V_k\cup C\not\fArrow V_i.
\]
Thus, at most $f$ incoming neighbors of $j$ belong to
$V_k\cup C$, and therefore,
at most $f$ values in $r_j[t]$ equal $\bfe_k$.

Since process $j$ does not know which of its incoming neighbors, if any, are faulty, it must
allow for the possibility that any of its $f$ incoming neighbors are faulty.
Let $A_k\subseteq V_k\cup C$, $k\neq i$, be the set containing all the incoming neighbors of process $j$
in $V_k\cup C$. Since $V_k\cup C\not\fArrow V_i$,
$|A_k|\leq f$; therefore, all
the processes in $A_k$ are {\em potentially} faulty.
Also,
by assumption, the values received from all fault-free processes equal
their input, and the values received from faulty processes in $F$ equal $\bfe_i$.
Thus, due to the validity condition, process $j$ must choose as its new state
a value that is in the convex hull of the set 
\[ S_k=\{\bfe_m~|~m\neq k, 0\leq m\leq p\}.\]
where $k\neq i$.
Since this observation is true for each $k\neq i$, it follows that
the new state $\bfv_j[t]$ must be a point in the convex hull of
\[ \cap_{1\leq k\leq p,~k\neq i}~ \HH(S_k).
\]
It is easy to verify that the above intersection only contains the point $\bfe_i$.
Therefore, $\bfv_j[t]=\bfe_i$. Thus, the state of process $j$ at the start of iteration $t+1$
equals $\bfe_i$. This concludes the induction.

The above result implies that the state of each fault-free process remains unchanged through the iterations. Thus, the state of any two fault-free processes
differs in at least one vector element by $2\epsilon$, precluding $\epsilon$-agreement.
\end{proof}

~\\
The above lemma demonstrates the necessity of Condition NC.
Necessary condition NC implies a lower bound on the number of processes $n=|\sv|$ in $G(\sv,\se)$,
as stated in the next lemma.

\comment{++++++++++++++
\paragraph{Lemma \ref{l_nc2}~~}
{\em
(Necessary condition NC2)
Suppose that the Byzantine vector consensus problem can be solved using an iterative algorithm in $G(\sv,\se)$.
 Then, for
$f>0$, each process must have in-degree $\geq (d+1)f+1$.
}

\begin{proof}
Suppose that Byzantine vector consensus problem can be solved using an iterative algorithm 
in $G(\sv,\se)$. Therefore, by Lemma \ref{l_nc1}, $G$ must satisfy NC.
Suppose that there exists a process $i$ that
has in-degree at most $(d+1)f$.
Define $V_0=\{i\}$.
Partition $N_i^-$ (i.e., incoming neighbors of process $i$) into $(d+1)$ sets,
each containing at most $f$ processes.
Name one of these $d+1$ sets as set $F$, and name the remaining
$d$ sets in the partition as $V_1,V_2,\cdots,V_d$.
Define $C=\sv-\cup_{k=0}^d V_k-F$. There are no links from $C$ to process $i$.
Since none of the $V_i$'s contain more than $f$ processes, and $C$ is empty,
this partition of $\sv$ does not satisfy Condition NC.
The proof is by contradiction.
\end{proof}
+++++++++++++++++++++++}

\begin{lemma}
\label{l_nc2}
Suppose that the Byzantine vector consensus problem can be solved using an iterative algorithm 
in $G(\sv,\se)$.
Then, $n\geq (d+2)f+1$.
\end{lemma}

\begin{proof}
Since the Byzantine vector consensus problem can be solved using an iterative algorithm 
in $G(\sv,\se)$, by Lemma \ref{l_nc1}, graph $G$ must satisfy Condition NC.
Suppose that $2\leq |\sv|=n\leq (d+2)f$. Then
there exists $p$, $1\leq p\leq d$, such that we can partition $\sv$ into 
sets $V_0,...,V_p,F$ 
such that for each $V_i$, $0<|V_i|\leq f$, and $|F|\leq f$.
Define $C=\emptyset$.
Since $|C\cup V_i|\leq f$ for each $i$, it is clear that this partition of $\sv$
cannot satisfy Condition NC. This is a contradiction.
\end{proof}

~


When $d=1$, the input at each process is a scalar. For the $d=1$ case, our prior work \cite{vaidya12podc}
yielded a tight necessary and sufficient condition for Byzantine consensus
to be achievable in $G(\sv,\se)$ using iterative algorithms.
For $d=1$,
the necessary condition stated in Lemma \ref{l_nc1} is equivalent to the necessary condition
in \cite{vaidya12podc}. We previously showed
 that, for $d=1$, the same condition is also sufficient
 \cite{vaidya12podc}.
However, in general, for $d>1$, Condition NC is not proved sufficient. Instead, we prove
the sufficiency of another condition stated in the next section.

\section{A Sufficient Condition}
\label{s_suff}



%

We now present Condition SC that is later proved to be sufficient for achieving
Byzantine vector consensus in graph $G(\sv,\se)$ using an iterative algorithm.

\paragraph{Condition SC:} 
{\em
For any partition $F,L,C,R$ of set $\scriptv$, such that $L$ and $R$ are both
non-empty, and $|F|\leq f$,
at least one of these conditions is true: ~$R\cup C\dfArrow L$, or ~$L\cup C\dfArrow R$.
} \\

Later in the paper we will present a Byzantine vector consensus algorithm named \algo~that is proved correct
in all graphs that saitsfy Condition SC. The proof will make use of Lemmas \ref{l_degree_sc}
and  
\ref{l_reduced} presented below.

\begin{lemma}
\label{l_degree_sc}
For $f>0$, if graph $G(\sv,\se)$ satisfies Condition SC, then
in-degree of each process in $\sv$ must be at least $(d+1)f+1$.
That is, for each $i\in\sv$, $|N_i^-|\geq (d+1)f+1$.
\end{lemma}

\noindent Lemma \ref{l_degree_sc} is proved in Appendix \ref{a_l_degree_sc}.

\begin{definition}
\label{def:reduced} {\bf Reduced Graph:}
For a given graph $G(\scriptv,\scripte)$ and $\sF\subset\scriptv$ such that $|\sF|\leq f$,
a graph $H(\scriptv_\sF,\scripte_\sF)$
is said to be a {\em reduced graph}, if: (i)
$\scriptv_\sF=\scriptv-\sF$, and (ii)
$\scripte_\sF$ is obtained by first removing from $\scripte$ all the links
incident on the processes in $\sF$, and {\em then} removing up to $df$ additional incoming
links at each process in $\scriptv_\sF$.
\end{definition}
Note that for a given $G(\scriptv,\scripte)$ and a given $\sF$,
multiple reduced graphs may exist (depending on the choice of the links
removed at each process).

\begin{lemma}
\label{l_reduced}
Suppose that graph $G(\sv,\se)$ satisfies Condition SC, and $\sF\subset \sv$.
Then, in any reduced graph $H(\sv_\sF,\se_\sF)$, there exists
a process that has a directed path to all the remaining processes in $\sv_\sF$.
\end{lemma}

\noindent Lemma \ref{l_reduced} is proved in Appendix \ref{a_l_reduced}.

\subsection{Algorithm \algo}

We will prove that,
if graph $G(\scriptv,\scripte)$ satisfies Condition SC, then
Algorithm \algo~presented below achieves Byzantine vector consensus.
Algorithm \algo~has the three-step structure described
in Section~\ref{s_structure}.

The proposed algorithm is based on the following
result by Tverberg \cite{perles07}.

\begin{theorem}
{\normalfont (Tverberg's Theorem \cite{perles07})}
\label{t_tverberg}
For any integer $f \geq 0$, and for every multiset $Y$
containing at least $(d+1)f+1$ points in ${\bf R}^d$, there exists a partition $Y_1,\cdots, Y_{f+1}$
of $Y$ into $f+1$ non-empty multisets such that $\cap_{l=1}^{f+1}\, \scripth(Y_l)\neq \emptyset$.
\end{theorem}
The points in $Y$ above need not be distinct \cite{perles07};
thus, the same point may occur multiple times in $Y$, and also in
each of its
subsets ($Y_l$'s) above.
The partition in Theorem \ref{t_tverberg} is called a {\em Tverberg partition}, and the points in
$\cap_{l=1}^{f+1}\, \scripth(Y_l)$ in Theorem \ref{t_tverberg} are called {\bf Tverberg points}.

\vspace*{8pt}\hrule
{\bf Algorithm \algo}
\vspace*{4pt}\hrule

\begin{list}{}{}
\item
Each iteration consists of three steps: {\em Transmit}, {\em Receive},
and {\em Update}:
\begin{enumerate}

\item {\em Transmit step:} Transmit current state $\bfv_i[t-1]$ on all outgoing edges.
\item {\em Receive step:} Receive values on all incoming edges. These values form
multiset $r_i[t]$ of size $|N_i^-|$.
(If a message is not received from some incoming neighbor,
then that neighbor must be faulty. In this case, the missing message value is
assumed to be $\bfe_0$ by default.
Recall that we assume a {\em synchronous} system.)

\item {\em Update step:}
Form a multiset $Z_i[t]$ using the steps below:
	\begin{itemize}
	\item Initialize $Z_i[t]$ as empty.
	\item Add to $Z_i[t]$, any one {\em Tverberg point} corresponding to {\em each} multiset $C\subseteq r_i[t]$ such that $|C|=(d+1)f+1$.
		Since $|C|=(d+1)f+1$, by Theorem \ref{t_tverberg}, such a Tverberg point exists.
	\end{itemize}
$Z_i[t]$ is a multiset; thus a single point may appear in $Z_i[t]$ more than once.
Note that $|Z_i[t]|~=~\nchoosek{|r_i[t]|}{(d+1)f+1}~\leq~ \nchoosek{n}{(d+1)f+1}$.
Compute new state $\bfv_i[t]$ as: \begin{eqnarray}
\bfv_i[t]~=~\frac{\bfv_i[t-1]+ \sum_{\bfz\in Z_i[t]}~ \bfz }{1+|Z_i[t]|}\label{e_algo_Z}
\end{eqnarray}
\end{enumerate}
\item
{\bf Termination:}
Each fault-free process terminates after completing $\terminate$ iterations, where $\terminate$
is a constant defined later in (\ref{e_delta_epsilon}).
The value of $\terminate$ depends on graph $G(\sv,\se)$, constants $U$ and $\mu$ defined earlier, and parameter $\epsilon$ of $\epsilon$-agreement.
\end{list}

\hrule

~

The proof of correctness of Algorithm \algo~makes use of a matrix representation of the
algorithm's behavior. Before presenting the matrix representation, we introduce some notations
and definitions related to matrices.

\subsection{Matrix Preliminaries}

We use boldface letters to denote matrices, rows of matrices, and their elements. For instance, $\bfA$ denotes a matrix, $\bfA_i$ denotes the $i$-th row of matrix $\bfA$, and $\bfA_{ij}$ denotes the element at the intersection of the $i$-th row and the $j$-th column of matrix $\bfA$.

\begin{definition}
\label{d_stochastic}
A vector is said to be stochastic if all its elements
are non-negative, and the elements add up to 1.
A matrix is said to be row stochastic if each row of the matrix is a
stochastic vector. 
\end{definition}
For matrix products, we adopt the ``backward'' product convention below, where $a \leq b$,
\begin{equation}
\label{backward}
\Pi_{\tau=a}^b \bfA[\tau] = \bfA[b]\bfA[b-1]\cdots\bfA[a]
\end{equation}
For a row stochastic matrix $\bfA$,
 coefficients of ergodicity $\delta(\bfA)$ and $\lambda(\bfA)$ are defined as
follows \cite{Wolfowitz}:
\begin{eqnarray*}
\delta(\bfA) & = &   \max_j ~ \max_{i_1,i_2}~ \| \bfA_{i_1\,j}-\bfA_{i_2\,j} \| \label{e_zelta} \\
\lambda(\bfA) & = & 1 - \min_{i_1,i_2} \sum_j \min(\bfA_{i_1\,j} ~, \bfA_{i_2\,j}) \label{e_lambda}
\end{eqnarray*}
\begin{claim}
\label{claim_zelta}
For any $p$ square row stochastic matrices $\bfA(1),\bfA(2),\dots, \bfA(p)$, 
\begin{eqnarray*}
\delta(\Pi_{\tau=1}^p \bfA(\tau)) ~\leq ~
 \Pi_{\tau=1}^p ~ \lambda(\bfA(\tau)).
\end{eqnarray*}
\end{claim}
Claim \ref{claim_zelta} is proved in \cite{Hajnal58}.
%
%
%
%
Claim \ref{c_lambda_bound} below follows directly from the definition of $\lambda(\cdotp)$. 
\begin{claim}
\label{c_lambda_bound}
If all the elements in any one column of matrix $\bfA$ are lower bounded by a constant $\gamma$,
then $\lambda(\bfA)\leq 1-\gamma$. That is, if $\exists g$, such that $\bfA_{ig}\geq \gamma$,
$\forall i$, then
$\lambda(\bfA)\leq 1-\gamma$.
\end{claim}

\subsection{Correctness of Algorithm \algo}

This section presents a key lemma, Lemma \ref{l_matrix}, that helps us in proving 
the correctness of Algorithm \algo. 
In particular,
Lemma \ref{l_matrix} allows us to use results for
non-homogeneous Markov chains to prove the correctness of Algorithm \algo.

Let $\sF$ denote the actual set of faulty processes in a given execution of Algorithm \algo.
Let $|\sF|=\psi$. Thus, $0\leq \psi\leq f$. Without loss of generality, suppose that 
processes $1$ through $(n-\psi)$ are fault-free, and if $\psi > 0$, processes $(n-\psi+1)$ through $n$ are faulty. 

In the analysis below, it is convenient to view the state of each process as a point in the
$d$-dimensional Euclidean space.
Denote by $\bfv[0]$ the column vector consisting of the initial states of the $(n-\psi)$ fault-free processes.
The $i$-th element of $\bfv[0]$ is $\bfv_i[0]$, the initial state of process $i$.
Thus, $\bfv[0]$ is a vector consisting of $(n-\psi)$ points in the $d$-dimensional Euclidean space.
 Denote by $\bfv[t]$, for $t \geq 1$, the column vector consisting of the states of the $(n-\psi)$ fault-free processes at the end of the $t$-th iteration. The $i$-th element of vector $\bfv[t]$ is state $\bfv_i[t]$.

\begin{lemma}
\label{l_matrix}
Suppose that graph $G(\sv,\se)$ satisfies Condition SC.
Then the state updates performed by the fault-free processes
in the $t$-th iteration ($t\geq 1$) of Algorithm \algo~can be
expressed as
\begin{eqnarray}
 \bfv[t] & = & \bfM[t]\,\bfv[t-1]
\label{e_matrix}
\end{eqnarray}
where $\matrixm[t]$ is a $(n-\psi) \times (n-\psi)$ row stochastic matrix
with the following property:
there exists a reduced graph $H[t]$,
and a constant $\beta$ ($0<\beta\leq 1$) that depends only on graph $G(\sv,\se)$, such that
 $$\matrixm_{ij}[t] ~ \geq ~ \beta$$
if $j=i$ or edge $(j,i)$ is in $H[t]$.
\end{lemma}
\begin{proof}
The proof is presented in Appendix \ref{a_l_matrix}.
\end{proof}

~

Matrix $\bfM[t]$ above is said to be a \underline{transition matrix}.
As the lemma states, $\bfM[t]$ is a row stochastic matrix.
The proof of Lemma \ref{l_matrix} shows how to identify a suitable row stochastic matrix $\bfM[t]$
for each iteration $t$. The matrix $\bfM[t]$ depends on $t$, as well as the behavior of the faulty processes.
$\bfM_i[t]$ is the $i$-th row of transition matrix $\bfM[t]$.
Thus, (\ref{e_matrix}) implies that
\[
\bfv_i[t] ~=~\bfM_i[t]\, \bfv[t-1]
\]
That is, the state of any fault-free process $i$ at the end of iteration $t$ can be expressed as a convex
combination of the state of just the fault-free processes at the end of iteration $t-1$.
Recall that vector $\bfv$ only includes the state of fault-free processes.

\begin{theorem}
\label{t_correct}
Algorithm \algo~satisfies the {\em termination}, {\em validity}
and $\epsilon$-agreement conditions.
\end{theorem}
\begin{proof}
Sections \ref{ss_validity}, \ref{ss_term} and \ref{ss_agreement}
provide the proof that Algorithm \algo~satisfies the three conditions
for Byzantine vector consensus.
This proof follows a structure used to prove correctness of other
consensus algorithms in our prior work \cite{vaidya_matrix_IABC,tseng13icdcn}.
\end{proof}

\subsection{Algorithm \algo~Satisfies the Validity Condition}
\label{ss_validity}

Observe that $\bfM[t+1]\left(\bfM[t]\bfv[t-1]\right)= \left(\bfM[t+1]\bfM[t]\right)\bfv[t-1]$. Therefore,
by repeated application of
(\ref{e_matrix}), we obtain for $t\geq 1$,
\begin{eqnarray}
\bfv[t] & = & \left(\,\Pi_{\tau=1}^t \matrixm[\tau]\,\right)\, \bfv[0]
\label{e_v_t}
\end{eqnarray}
Since each $\matrixm[\tau]$ is row stochastic, the matrix product
$\Pi_{\tau=1}^t \matrixm[\tau]$ is also a row stochastic matrix.
Recall that vector $\bfv$ only includes the state of fault-free processes.
Thus, (\ref{e_v_t}) implies that the state of each fault-free process $i$ at the end of iteration $t$ can be expressed as a convex
combination of the initial state of the fault-free processes.
Therefore, the validity condition is satisfied.

\subsection{Algorithm \algo~Satisfies the Termination Condition}
\label{ss_term}

Algorithm \algo~stops after a finite number ($\terminate$) of iterations,
where $\terminate$ is a constant that depends only on $G(\sv,\se)$,
$U$, $\mu$ and $\epsilon$.
Therefore, trivially, the algorithm
satisfies the termination condition.
Later, using (\ref{e_delta_epsilon}) we define a suitable value for
$\terminate$.

\subsection{Algorithm \algo~Satisfies the $\epsilon$-Agreement Condition}
\label{ss_agreement}


The proof structure below is derived from our previous work
wherein we proved the correctness of an iterative algorithm for {\em scalar} Byzantine consensus (i.e.,
the case of $d=1$) \cite{vaidya_matrix_IABC} and its generalization to
a broader class of fault sets \cite{tseng13icdcn}.

Let $R_{F}$ denote the set of all the reduced graph of $G(\scriptv, \scripte)$
corresponding to fault set $F$. Thus,
$R_{\sF}$ is the set of all the reduced graph of $G(\scriptv, \scripte)$
corresponding to actual fault set $\sF$.
Let \[ r = \max_{|F|\leq f}\, |R_{F}|.\] $r$ depends only on $G(\sv,\se)$ and $f$, and it is finite.
Note that $|R_\sF|\leq r$.

For each reduced graph $H\in R_\sF$, define
connectivity matrix $\matrixh$ as follows,
where $1\leq i,j\leq n-\psi$:
\begin{itemize}
\item $\matrixh_{ij}=1$ if either $j=i$, or 
	edge $(j, i)$ exists in reduced graph $H$.
\item $\matrixh_{ij}=0$, otherwise.
\end{itemize}
Thus, the non-zero elements of row $\matrixh_i$ correspond to the incoming links at process $i$ in the reduced graph $H$, 
and the self-loop at process $i$. Observe that 
$\matrixh$ has a non-zero diagonal.

\begin{lemma}
\label{lemma:non-zero}
For any $\graphh \in R_{F}$, and any $k\geq n-\psi$, matrix product $\matrixh^k$ has at
least one non-zero column (i.e., a column with all elements non-zero).
\end{lemma}
\begin{proof}
Each reduced graph contains $n-\psi$ processes because the fault set $\sF$ contain $\psi$ processes.
By Lemma \ref{l_reduced}, at least one process in the reduced graph, say process $p$, has directed paths to all the 
processes in the reduced graph $H$.
Element $\matrixh^k_{jp}$ of matrix product $\matrixh^k$ is 1 if and only if process
$p$ has a directed path to process $j$ containing at most $k$ edges; each of these directed
paths must contain less than $n-\psi$ edges, because the number of processes
in the reduced graph is $n-\psi$. 
Since $p$ has directed paths to all the processes,
it follows that, when $k\geq n-\psi$,
all the elements in the $p$-th column of $\matrixh^k$
must be non-zero.
\end{proof}

For matrices $\bfA$ and $\bfB$ of identical dimensions, we say that
$\bfA\leq \bfB$ if and only if $\bfA_{ij}\leq \bfB_{ij}$, $\forall i,j$.
Lemma \ref{l_cm} relates the transition matrices with the connectivity matrices.
Constant $\beta$ used in the lemma below was introduced in Lemma \ref{l_matrix}.

\begin{lemma}
\label{l_cm}
For any $t \geq 1$, there exists a reduced graph $\graphh[t] \in R_{\sF}$ such
 that $\beta {\normalfont\bf\matrixh[t] \leq \matrixm}[t]$,
where $\matrixh[t]$ is the connectivity matrix for $H[t]$.
\end{lemma}
\begin{proof}
Appendix \ref{a_l_cm} presents the proof.
\end{proof}

\begin{lemma}
\label{l_product_H}
At least one column in the matrix product
$
\Pi_{t=u}^{u+r(n-\psi)-1} \, \bfH[t]
$
is non-zero.
\end{lemma}
\begin{proof}
Since $\Pi_{t=u}^{u+r(n-\psi)-1} \, \bfH[t]$ is a product of $r(n-\psi)$ connectivity matrices
corresponding to the reduced graphs
in $R_\scriptf$, and $|R_\sF|\leq r$,
connectivity matrix
corresponding to at least one reduced graph
in $R_\scriptf$, say matrix $\bfH_*$\,, will appear in the above
product at least $n-\psi$ times.

By Lemma \ref{lemma:non-zero}, $\bfH_*^{n-\psi}$ contains a non-zero
column; say the $p$-th column of $\bfH_*$ is non-zero.
Also, by definition, all the connectivity matrices ($\bfH[t]$) have a non-zero diagonal.
These two observations together imply that the $p$-th column in the product 
$\Pi_{t=u}^{u+r(n-\psi)-1} \, \bfH[t]$ is non-zero.\footnote{The product
$\Pi_{t=z}^{z+r(n-\psi)-1} \, \bfH[t]$ can be viewed
as the product of $(n-\psi)$ instances of $\bfH_*$ ``interspersed'' with
matrices with non-zero diagonals. 
}

\end{proof}

Let us now define a sequence of matrices $\bfQ(i)$, $i\geq 1$, such that
each of these matrices is a product of $r(n-\psi)$ of the
$\bfM[t]$ matrices. Specifically,
\begin{eqnarray}
\bfQ(i) &=& \Pi_{t=(i-1)r(n-\psi)+1}^{ir(n-\psi)} ~ \bfM[t]
\label{e_Q_i}
\end{eqnarray}
From (\ref{e_v_t}) and (\ref{e_Q_i})
observe that
\begin{eqnarray}
\bfv[kr(n-\psi)] & = & \left(\, \Pi_{i=1}^k ~ \bfQ(i) \,\right)~\bfv[0]
\end{eqnarray}

\begin{lemma}
\label{l_Q}
For $i\geq 1$, $\bfQ(i)$ is a row stochastic matrix,
and \[ \lambda(\bfQ(i))\leq 1-\beta^{r(n-\psi)}.\]
\end{lemma}
\begin{proof}
$\bfQ(i)$ is a product of row stochastic matrices ($\bfM[t]$); therefore,
$\bfQ(i)$ is row stochastic.
From Lemma \ref{l_cm}, for each $t\geq 1$,
\[
\beta \, \bfH[t] ~ \leq ~ \bfM[t]
\]
Therefore, 
\[
\beta^{r(n-\psi)} ~ \Pi_{t=(i-1)r(n-\psi)+1}^{ir(n-\psi)} ~ \bfH[t] ~ \leq 
~ \Pi_{t=(i-1)r(n-\psi)+1}^{ir(n-\psi)} ~ \bfM[t] ~ =
~ \bfQ(i)
\]
By using $u=(i-1)r(n-\psi)+1$ in Lemma \ref{l_product_H},
we conclude that the matrix product on the left side
of the above inequality contains a non-zero column. Therefore, since $\beta>0$, $\bfQ(i)$ on the
right side of the inequality also contains
a non-zero column.

Observe that $r(n-\psi)$ is finite, and hence, $\beta^{r(n-\psi)}$
is non-zero. Since the non-zero terms in $\bfH[t]$ matrices are all 1,
the non-zero elements in $\Pi_{t=(i-1)r(n-\psi)+1}^{ir(n-\psi)} \bfH[t]$
must each be $\geq$ 1. Therefore, there exists a non-zero column in $\bfQ(i)$
with all the elements in the column being $\geq \beta^{r(n-\psi)}$.
Therefore, by Claim \ref{c_lambda_bound}, $\lambda(\bfQ(i))\leq 1-\beta^{r(n-\psi)}$.
\end{proof}
~\\
~\\

Let us now continue with the proof of $\epsilon$-agreement.
Consider the coefficient of ergodicity $\delta(\Pi_{i=1}^t \bfM[i])$.
\begin{eqnarray}
 \delta(\Pi_{i=1}^t \bfM[i])
&=& \delta\left(
\left( \Pi_{i=\lfloor\frac{t}{r(n-\psi)}\rfloor r(n-\psi)+1}^t \bfM[i] \right)
\left(\Pi_{i=1}^{\lfloor\frac{t}{r(n-\psi)}\rfloor} \bfQ(i)\right)
\right) \mbox{~~~~ by definition of $Q(i)$} \nonumber \\ 
& \leq &  \lambda 
\left( \Pi_{i=\lfloor\frac{t}{r(n-\psi)}\rfloor r(n-\psi)+1}^t \bfM[i] \right)
\, \Pi_{i=1}^{\lfloor\frac{t}{r(n-\psi)}\rfloor} \lambda(\bfQ(i)) 
\mbox{~~~~ by Claim \ref{claim_zelta}} \nonumber \\
& \leq &  \Pi_{i=1}^{\lfloor\frac{t}{r(n-\psi)}\rfloor} \lambda(\bfQ(i)) \mbox{~~~~because $\lambda(.)\leq 1$} 
\nonumber \\
& \leq & \left(1-\beta^{r(n-\psi)}\right)^{\lfloor\frac{t}{r(n-\psi)}\rfloor}
\mbox{~~~~by Lemma \ref{l_Q}} \nonumber \\
& \leq & \left(1-\beta^{rn}\right)^{\lfloor\frac{t}{rn}\rfloor}
\mbox{~~~~because $0<\beta\leq 1$ and $0\leq \psi<n$.}
\label{e_delta_M}
\end{eqnarray}
Observe that the upper bound on right side of (\ref{e_delta_M}) depends only on
graph $G(\sv,\se)$ and $t$, and is independent of the input vectors, the fault set $\sF$,
and the behavior of the faulty processes.
Also, the upper bound on the right side of (\ref{e_delta_M}) is a
non-increasing function of $t$.
Define
$\terminate$
as the smallest positive integer $t$ for which the right hand side of (\ref{e_delta_M}) is
smaller than $\frac{\epsilon}{n\max(\|U\|,\|\mu\|)}$,
where $\|x\|$ denotes the absolute value of real number $x$. Thus,
\begin{eqnarray}
 \delta(\Pi_{i=1}^{\terminate} \, \bfM[i])
~ \leq ~ \left(1-\beta^{rn}\right)^{\left\lfloor\frac{\terminate}{rn}\right\rfloor}
~ < ~ \frac{\epsilon}{n\max(\|U\|,\|\mu\|)}
\label{e_delta_epsilon}
\end{eqnarray}
Recall that $\beta$ and $r$ depend only on $G(\sv,\se)$.
Thus, $\terminate$ depends only on graph $G(\sv,\se)$,
and constants $U$, $\mu$ and $\epsilon$. 

Recall that $\Pi_{i=1}^t\bfM[i]$ is a $(n-\psi)\times (n-\psi)$ row stochastic matrix.
Let $\bfM^*=\Pi_{i=1}^t\bfM[i]$.
From (\ref{e_v_t}) we know that
state
$\bfv_j[t]$ of any fault-free process $j$ is obtained as the product of the $j$-th row
of $\Pi_{i=1}^{t} \, \bfM[i]$ and $\bfv[0]$.  
That is, $\bfv_j[t]=\bfM^*_j\bfv[0]$.

Recall that $\bfv_j[t]$ is a $d$-dimensional vector. Let us denote
the $l$-th element of $\bfv_j[t]$ as $\bfv_j[t](l)$, $1\leq l\leq d$.
Also, by $\bfv[0](l)$, let us denote a vector consisting of
the $l$-th elements of $\bfv_i[0], \forall i$.
Then by the definitions of $\delta(.)$, $U$ and $\mu$,
for any two fault-free processes $j$ and $k$,
we have
\begin{eqnarray}
\|\bfv_j[t](l)-\bfv_k[t](l)\| & = & 
\|\bfM^*_j\bfv[0](l)
- \bfM^*_k\bfv[0](l)\| \\
&=& \|\sum_{i=1}^{n-\psi} \bfM^*_{ji}\bfv_i[0](l)
- \sum_{i=1}^{n-\psi} \bfM^*_{ki}\bfv_i[0](l)\| \\
&=& \|\sum_{i=1}^{n-\psi} \left(\bfM^*_{ji}- \bfM^*_{ki}\right)\bfv_i[0](l)\| \\
&\leq& \sum_{i=1}^{n-\psi} \|\bfM^*_{ji}- \bfM^*_{ki}\|\,\|\bfv_i[0](l)\| \\
&\leq& \sum_{i=1}^{n-\psi} \delta(\bfM^*)\|\bfv_i[0](l)\| \\
&\leq& (n-\psi) \delta(\bfM^*) \max(\|U\|,\|\mu\|) \\
&\leq & (n-\psi)\max(\|U\|,\|\mu\|) \, \delta(\Pi_{i=1}^t \bfM[i]) \nonumber \\
& \leq & n \max(\|U\|,\|\mu\|)  \, \delta(\Pi_{i=1}^t \bfM[i])  \mbox{~~~~ because $0\leq \psi<n$} \label{e_d}
\end{eqnarray}
Therefore, by (\ref{e_delta_epsilon}) and (\ref{e_d}),
\begin{eqnarray} \|\bfv_i[\terminate](l)-\bfv_j[\terminate](l)\| &<& \epsilon, ~~~~~~1\leq l\leq d.\label{e_e}
\end{eqnarray}
The output of a fault-free process equals its state at termination (after
$\terminate$ iterations).
Thus, (\ref{e_e}) implies that Algorithm \algo~satisfies the $\epsilon$-agreement condition. \\

\comment{+++++++++++++++++++++++++++
\section{Correctness of Algorithm \algo~Under Condition SC2}
\label{sec:sufficiency}

in this section, we
prove that Algorithm \algo satisfies {\em validity}, {\em $\epsilon$-agreement}
and {\em termination} conditions, respectively, provided that $G(\scriptv,\scripte)$
satisfies condition SC2.

\section{New Necessary condition}

Consider partition $L,R,C,F$. Consider set $B$, defined as the processes
in $L,R$ that have links
from $R,L$, respectively, and all reachable processes in $L+C$ and $R+C$, respectively.
It should not be be possible to color processes in $B$
with at most $d+1$ colors such that
each process in $B$ has at most $f$ neighbors in (each color+set C),
except its own color.

This is equivalent to PODC'12 for the case when $d=1$.

Necessity for general $d$: trivial

Sufficiency?
++++++++++++++++++++++++++}

\section{Summary}
\label{s_summary}

This paper addresses {\em Byzantine vector consensus} (BVC), wherein the input at each process is a $d$-dimensional vector of reals, and each process is expected to decide on a {\em decision vector} that is in the {\em convex hull} of the input vectors at the fault-free processes  \cite{mendes13stoc,vaidya13podc}.
We address a particular class of
{\em iterative} algorithms in {\em incomplete} graphs, and prove a necessary condition (NC), and a sufficient condition (SC),
for the graphs to be able to solve the vector consensus problem iteratively.
This paper extends our past work on {\em scalar} consensus (i.e., $d=1$) in incomplete graphs in presence of Byzantine faults \cite{vaidya12podc,vaidya_matrix_IABC},
which yielded an
exact characterization of graphs in which the problem is solvable for $d=1$. However, the necessary
condition NC presented in the paper for {\em vector} consensus does not match with the sufficient
condition SC.
We hope that this paper will motivate further work on identifying
the tight sufficient condition.


\begin{thebibliography}{1}

\bibitem{dag_decomposition}
S.~Dasgupta, C.~Papadimitriou, and U.~Vazirani.
\newblock {\em Algorithms}.
\newblock McGraw-Hill Higher Education, 2006.

\bibitem{Hajnal58}
J.~Hajnal.
\newblock Weak ergodicity in non-homogeneous markov chains.
\newblock In {\em Proceedings of the Cambridge Philosophical Society},
  volume~54, pages 233--246, 1958.

\bibitem{mendes13stoc}
H.~Mendes and M.~Herlihy.
\newblock Multidimensional approximate agreement in byzantine asynchronous
  systems.
\newblock In {\em 45th ACM Symposium on the Theory of Computing (STOC)}, June
  2013.

\bibitem{perles07}
M.~A. Perles and M.~Sigron.
\newblock A generalization of {Tverberg}'s theorem, 2007.
\newblock CoRR, http://arxiv.org/abs/0710.4668.

\bibitem{tseng13icdcn}
L.~Tseng and N.~H. Vaidya.
\newblock Iterative approximate byzantine consensus under a generalized fault
  model.
\newblock In {\em International Conference on Distributed Computing and
  Networking ({ICDCN})}, January 2013.

\bibitem{tseng13optconvex}
L. Tseng and N. H. Vaidya,
Byzantine Convex Consensus: An Optimal Algorithm, 2013.
CoRR, http://arxiv.org/abs/1307.1332.

\bibitem{vaidya_matrix_IABC}
N.~H. Vaidya.
\newblock Matrix representation of iterative approximate byzantine consensus in
  directed graphs.
\newblock {\em CoRR http://arxiv.org/abs/1203.1888}, March 2012.

\bibitem{vaidya13podc}
N.~H. Vaidya and V.~K. Garg.
\newblock Byzantine vector consensus in complete graphs.
\newblock In {\em ACM Symposium on Principles of Distributed Computing
  ({PODC})}, July 2013.

\bibitem{vaidya12podc}
N.~H. Vaidya, L.~Tseng, and G.~Liang.
\newblock Iterative approximate byzantine consensus in arbitrary directed
  graphs.
\newblock In {\em ACM Symposium on Principles of Distributed Computing
  ({PODC})}, July 2012.

\bibitem{Wolfowitz}
J.~Wolfowitz.
\newblock Products of indecomposable, aperiodic, stochastic matrices.
\newblock In {\em Proceedings of the American Mathematical Society}, pages
  733--737, 1963.

\end{thebibliography}

\appendix

\section{Proof of Lemma \ref{l_degree_sc}}
\label{a_l_degree_sc}

\paragraph{Lemma \ref{l_degree_sc}~}
{\em
For $f>0$, if graph $G(\sv,\se)$ satisfies Condition SC, then
in-degree of each process in $\sv$ must be at least $(d+1)f+1$.
That is, for each $i\in\sv$, $|N_i^-|\geq (d+1)f+1$.
\\~\\ }
\begin{proof}
The proof is by contradiction.
As per the assumption in the lemma, $f>0$, and graph $G(\sv,\se)$ satisfies condition SC. 

Suppose that some process $i$ 
has in-degree at most $(d+1)f$.
Define $L=\{i\}$, and $C=\emptyset$.
Partition the processes in $\sv-\{i\}$ into sets
$R$ and $F$ such that $|F|\leq f$, $|F\cap N_i^-|\leq f$ and $|R\cap N_i^-|\leq df$.
Such sets $R$ and $F$ exist because in-degree of process $i$ is at most $(d+1)f$.
$L,R,C,F$ thus defined form a partition of $\sv$.

Now, $f>0$ and $d\geq 1$, and $|L\cup C|=1$.
Thus, there can be at most $1$ link from $L\cup C$ to any process
in $R$, and $1\leq df$.
Therefore, $L\cup C\not\dfArrow R$.
Also, because $C=\emptyset$, $|(R\cup C)\cap N_i^-|=|R\cap N_i^-|\leq df$.
Thus, there can be at most $df$ links from $R\cup C$ to process $i$, which is the
only process in $L=\{i\}$.
Therefore, $R\cup C\not\dfArrow L$.
Thus, the above partition of $\sv$ does not satisfy
Condition SC. This is a contradiction.
\end{proof}

\section{Proof of Lemma \ref{l_reduced}}
\label{a_l_reduced}

Before presenting the proof of Lemma \ref{l_reduced}, we introduce some terminology.

\begin{definition}
\label{def:decompose}
{\bf Graph decomposition:}
Let $H$ be a directed graph. Partition graph $H$ into strongly connected components,
$H_1,H_2,\cdots,H_h$, where $h$ is a non-zero integer dependent on graph $H$,
such that
\begin{itemize}
\item every pair of processes {\bf within} the same strongly connected component has directed
paths in $H$ to each other, and
\item for each pair of processes, say $i$ and $j$, that belong to
two {\bf different} strongly connected components, either $i$ does not have a
directed path to $j$ in $H$, or $j$ does not have a directed path to $i$ in $H$.
\end{itemize}
Construct a graph $H^*$ wherein each strongly connected component $H_k$ above is represented
by vertex $c_k$, and there is an edge from vertex $c_k$ to vertex $c_l$ only if
the processes in $H_k$ have directed paths in $H$ to the processes in $H_l$.
\end{definition}
It is known that the decomposition
graph $H^*$ is a directed {\em acyclic} graph \cite{dag_decomposition}.

\begin{definition}
{\bf Source component}:
Let $H$ be a directed graph, and let $H^*$ be its decomposition as per
Definition~\ref{def:decompose}. 
Strongly connected component $H_k$ of $H$ is said to be a {\em source component}
if the corresponding vertex $c_k$ in $H^*$ is \underline{not} reachable from any
other vertex in $H^*$. 
\end{definition}

\paragraph{Lemma \ref{l_reduced}} {\em
Suppose that graph $G(\sv,\se)$ satisfies Condition SC, and $\sF\subset \sv$.
Then, in any reduced graph $H(\sv_\sF,\se_\sF)$, there exists
a process that has a directed path to all the remaining processes in $\sv_\sF$.
\\~\\
}
%
%
\begin{proof}
Suppose that graph $G(\sv,\se)$ satisfies Condition SC.
We first prove that the reduced graph $H(\sv_\sF,\se_\sF)$ contains exactly one
{\em source component}.

Since $|\sF|<|\scriptv|$, reduced graph 
$H(\sv_\sF,\se_\sF)$ contains at least one process;
therefore, at least one
source component must exist in the reduced graph $H$.
(If $H$ consists of a single strongly connected component,
then that component is trivially a source component.)

So it remains to prove that $H(\sv_\sF,\se_\sF)$ cannot
contain more than one source component.
The proof is by contradiction.

Suppose that the decomposition of
$H(\scriptv_\sF,\scripte_\sF)$ 
contains at least two source components.
Let the sets of processes in two such source components of the reduced graph $H$
be denoted as $L$ and $R$, respectively. Let $C=\scriptv_\sF-L-R=\sv-\sF-L-R$.
Observe that $\sF,L,C,R$ form a partition of the processes in $\scriptv$.
Since $L$ is a source component in the reduced graph $H(\sv_\sF,\se_\sF)$, 
there are no directed links in $\scripte_\sF$ from any process in
$C\cup R$ to the processes in $L$.
Similarly, since $R$ is a source component in the reduced graph $H$,
there are no directed links in $\scripte_\sF$ from any process in $L\cup C$ to
the processes in $R$.
These observations, together with the manner in which $\scripte_\sF$
is defined, imply that (i) there are at most $df$ links in $\scripte$ from
the processes in $C\cup R$ to each process in $L$, and
(ii) there are at most $df$ links in $\scripte$ from
the processes in $L\cup C$ to each process in $R$.
Therefore, in graph $G(\scriptv,\scripte)$, $C\cup R\not\dfArrow L$
and $L\cup C\not\dfArrow R$. This violates Condition SC, resulting in a contradiction.
Thus, we have proved that $H(\sv_\sF,\se_\sF)$ must contain exactly one source component. 

Consider any process in the unique source component, say process $s$. By definition
of a strongly connected component,
process $s$ has directed paths to all the processes in the source component
using the edges in $\se_\sF$. 
Also, by the uniqueness of the source component, all other strongly connected components in $H$
(if any exist) are {\bf not} source components, and hence reachable from the source component the edges in $\se_\sF$. 
Therefore, process $s$ also has paths to all the processes in $\sv_\sF$ that
are outside the source component as well.
Therefore, process $s$ has paths to all the process in $\sv_\sF$.
This proves the lemma.
\end{proof}

~

The above proof shows that, if Condition SC is true, then
each reduced graph contains exactly one source component.
It is also possible to show that, if 
each reduced graph $H$ contains exactly one source component,
then Condition SC is satisfied.

\comment{+++++++++++++++++++++++++++++++++++++++++++++++++++++++++++++++++++++++++++++++++++++
++++++++++++++++  NOTES +++++++++++++++++++++++++

Necessary condition:
L, R, C, F
Color such that each boundary processes has at most f incoming neighbors of any color on the "other" side.
Then there must be a boundary process with neighbors of d+1 colors (remove this? on the other side).
Necessity is obvious (d colors don't suffice) -- choose points such that each set of d+1 points is affinely independent, and then assign one point to each color.

++++++

Necessary condition:
L,R,C,F
Color processes in L with d1>0 colors, and processes in R with d2>0 colors, such that
d1+d2 <= d+1. Then there must be a process with f+1 neighbors of same color on the
other side.
++++++++++++++++++++++++++++++++++++++++++++++++++++++}

\section{Proof of Lemma \ref{l_matrix}}
\label{a_l_matrix}

Recall that $\sF$ is actual set of faults in a given execution of the proposed
algorithm, and $|\sF|=\psi$. As noted before, without loss of generality,
we assume that processes 1 through $n-\psi$ are fault-free, and rest are faulty.
To simplify the terminology, the definition below assumes a certain iteration index $t\geq 1$.
\begin{definition}
\label{d_valid}
{\bf $\chi$-dependence:}
For a constant $\chi$, $0\leq \chi\leq 1$,
a point $\bfr$ in the convex hull of  $\{ \bfv_i[t-1]\,|\, 1\leq i\leq n-\psi\}$
is said to be \underline{$\chi$-dependent on process $k$} if there \underline{exist}
constants $\alpha_i$, $1\leq i\leq n-\psi$, such that $0\leq \alpha_i\leq 1$,
$\sum_{1\leq k\leq n-\psi}~\alpha_i = 1$, and
$$\alpha_k\geq \chi$$
such that
\[
\bfr = \sum_{1\leq i\leq n-\psi}~ \alpha_i\,\bfv_i[t-1]
\]
$\alpha_i$ is said to be the {\bf weight} of $\bfv_i[t-1]$ in the above convex combination.
\end{definition}

\begin{lemma}
\label{l_chi}
Let $P\subseteq\sv-\sF$ be a non-empty subset of fault-free processes.
Any point $\bfr$ in the convex hull of  $\{ \bfv_j[t-1]\,|\, j\in P\}$
is $\frac{1}{n}$-dependent on at least one fault-free process in $P$.
\end{lemma}
\begin{proof}
Recall that we assume processes 1 through $n-\psi$ to be fault-free,
and the remaining processes to be faulty.
Any point $\bfr$ in the convex
hull of the state of fault-free processes in $P$ can be written as their convex combination.
Since there are at most $n$ fault-free processes in $P$, and their weights in the convex
combination add to 1, at least one of the weights must be $\geq\frac{1}{n}$, proving
the lemma.
\end{proof}

\begin{definition}
Points in multiset $R$ are said to be collectively $\chi$-dependent on processes in
set $P$, if for each $p\in P$, there exists $\bfr\in R$ such that $\bfr$ is $\chi$-dependent
on $p$.
\end{definition}

\paragraph{Lemma \ref{l_matrix}} {\em
Suppose that graph $G(\sv,\se)$ satisfies Condition SC.
Then the state updates performed by the fault-free processes
in the $t$-th iteration ($t\geq 1$) of Algorithm \algo~can be
expressed as
\begin{eqnarray}
 \bfv[t] & = & \bfM[t]\,\bfv[t-1]
\label{e_matrix_appendix}
\end{eqnarray}
where $\matrixm[t]$ is a $(n-\psi) \times (n-\psi)$ row stochastic matrix
with the following property:
there exists a reduced graph $H[t]$,
and a constant $\beta$ ($0<\beta\leq 1$) that depends only on graph $G(\sv,\se)$, such that
 $$\matrixm_{ij}[t] ~ \geq ~ \beta$$
if $j=i$ or edge $(j,i)$ is in $H[t]$.
}
~\\

\noindent
\begin{proof}
We consider the case of $f=0$ separately from $f>0$.
\begin{itemize}
\item {\bf $f=0$:}
When $f=0$, all the processes are fault-free
(i.e., $\sF=\emptyset$), and $(d+1)f+1=1$. In this case, there is only one reduced
graph, which is identical to $G(\sv,\se)$. Because $(d+1)f+1=1$,
each multiset $C$ used in the {\em Update step} of
Algorithm \algo~to compute multiset $Z_i[t]$ contains value received from exactly one incoming
neighbor. (When $f=0$, and Condition SC holds true, it is possible
that exactly one process in the graph has no incoming neighbors.
If some process $j$ has no incoming neighbors, then $Z_j[t]=\emptyset$.)

For $C=\{\bfx\}$, that is, $C$ containing a single point $\bfx$, the Tverberg point for $f=0$ is $\bfx$ as well.
Thus, $|Z_i[t]|=|N_i^-|$, and $\bfv_i[t]$ is simply the average of $\bfv_i[t-1]$ and the values received from
all the incoming neighbors of $i$, which are necessarily fault-free (because $f=0$).
 Thus, $\bfv_i[t]$ is a convex combination of the elements of
$\bfv[t-1]$, wherein the weight assigned to each $j$ such that $j=i$ or $(j,i)\in\se$ is
$\frac{1}{1+|N_i^-|}$. Since $1+|N_i^-|\leq n$, by defining $\beta=\frac{1}{n}$, the statement of the lemma follows.

\item {\bf $f>0$:}
Consider a fault-free process $i$.
Suppose that the number of faulty incoming neighbors of process $i$ is $f_i\leq f$.
When Condition SC holds, and $f>0$, as shown in Lemma \ref{l_degree_sc}, 
each process has an in-degree of at least $(d+1)f+1$.
Therefore, for some integer $\kappa\geq 1$,
let $$|r_i[t]|=|N_i^-|=(d+1)f+\kappa=df+(f-f_i+\kappa)+f_i.$$
Recall that
the {\em Update step} of Algorithm \algo~enumerates suitable subsets $C$ of
multiset $r_i[t]$, and picks one {\em Tverberg point} corresponding to each such $C$.
By an inductive argument we will identify $\kappa$ such subsets $C_1,C_2,\cdots,C_\kappa$,
such that the Tverberg points added to $Z_i[t]$ corresponding to those
$\kappa$ subsets are collectively dependent on at least $(f+1)-f_i=f-f_i+\kappa$
fault-free incoming neighbors of process $i$.
Let the Tverberg point added corresponding to $C_j$ be denoted as $\bfz_j$.
\begin{itemize}
\item
Consider a subset $C_1$ of $r_i[t]$ such that $|C_1|=(d+1)f+1$.
A Tverberg point $\bfz_1$ for $C_1$ is added to $Z_i[t]$ in the {\em Update step}.
By the definition of a Tverberg point, there exists a partition $V_1,V_2,\cdots,V_{f+1}$
of multiset $C_1$, wherein each $V_j$ is non-empty, such that
\[
\bfz_1\in\cap_{1\leq j\leq f+1}\HH(V_j)
\]
Since process $i$ has at most $f_i$ faulty incoming neighbors,
at most $f_i$ values in $C_1$ are received from faulty neighbors.
Thus, at least $(f+1)-f_i=f-f_i+1$ of the subsets in the above partition contain values received from
only fault-free neighbors of process $i$.
For each such fault-free $V_k$, $\bfz_1\in\HH(V_k)$, and by Lemma \ref{l_chi}, $\bfz_1$ must 
be $\frac{1}{n}$-dependent on at least one fault-free neighbor of $i$ whose value
is included in $V_k$. Since the $V_j$'s form a partition, this
implies that there are $f-f_i+1$ distinct fault-free incoming neighbors
of $i$ on which $\bfz_1$ is $\frac{1}{n}$-dependent.
Let $\{p_1,p_2,\cdots,p_{f-f_i+1}\}$ denote $f-f_i+1$ distinct
incoming fault-free neighbors of $i$ on which $\bfz_1$ is $\frac{1}{n}$-dependent.
Note that $\{p_1,p_2,\cdots,p_{f-f_i+1}\}$ is a subset of the processes whose
values are included in $C_1$.

If $\kappa=1$, then we have already identified the subsets $C_1,\cdots,C_\kappa$
as desired. If $\kappa>1$, then we inductively identify the remaining $C_i$'s below.

\item
Suppose that $\kappa>1$, and that we have identified subsets $C_1,\cdots,C_\nu$, where $1\leq \nu<\kappa$
such that 
$\{\bfz_1,\bfz_2,\cdots,\bfz_\nu\}$
are collectively $\frac{1}{n}$-dependent on $f-f_i+\nu$ distinct incoming fault-free neighbors
of process $i$ that form
the set $\{p_1,p_2,\cdots,p_{f-f_i+\nu}\}$.
(The previous item proved the correctness of this assumption for $\nu=1$.)

Pick a subset $$C_{\nu+1}\subseteq r_i[t]-\cup_{j=1}^{\nu}\,\{\bfv_{p_j}[t-1]\},$$ such that $|C_{\nu+1}|=(d+1)f+1$.
In other words, $C_{\nu+1}$ does not contain values received from 
the $\nu$ neighbors in $\{p_1,p_2,\cdots,p_\nu\}$ (these neighbors are fault-free
by definition, and hence correctly send their state). 
Such a set $C_{\nu+1}$ must exist because $1\leq \nu<\kappa$,
and $|N_i^-|=(d+1)f+\kappa\geq (d+1)f+1+\nu$.

Note that $r_i[t]$ is a multiset, and $r_i[t]-\cup_{j=1}^{\nu}\,\{\bfv_{p_j}[t-1]\}$ is a multiset
as well. As an example,
if a value appears in $r_i[t]$ three time, and appears only once in $\cup_{j=1}^{\nu}\,\{\bfv_{p_j}[t-1]\}$,
then that value will appear twice in $r_i[t]-\cup_{j=1}^{\nu}\,\{\bfv_{p_j}[t-1]\}$.

By an argument similar to the previous item, we can show that the Tverberg point $\bfz_{\nu+1}$
corresponding to $C_{\nu+1}$ must be $\frac{1}{n}$-dependent on at least $f-f_i+1$
faulty-free processes from whom the values in $C_{\nu+1}$ are received.
By definition of $C_{\nu+1}$, processes $p_1,\cdots,p_\nu$ are not among these
$f-f_i+1$ processes. Thus, among these $f-f_i+1$ fault-free processes,
there exists at least one fault-free incoming neighbor of $i$ that is not included
in $\{p_1,p_2,\cdots,p_{f-f_i+\nu}\}$. Let us denote one such neighbor as $p_{f-f_i+\nu+1}$.
Thus, we have identified set $\{p_1,p_2,\cdots,p_{f-f_i+\nu+1}\}$ consisting
of $f-f_i+\nu+1$ fault-free incoming neighbors of process $i$ such that the
points in $\{\bfz_1,\bfz_2,\cdots,\bfz_{\nu+1}\}$ are collectively $\frac{1}{n}$-dependent on
$\{p_1,p_2,\cdots,p_{f-f_i+\nu+1}\}$.
\end{itemize}
Note that $\{\bfz_1,\bfz_2,\cdots,\bfz_\kappa\}\subseteq Z_i[t]$.
The above argument inductively proves that there exist $f-f_i+\kappa$ incoming fault-free neighbors
of process $i$, forming set $\{p_1,p_2,\cdots,p_{f-f_i+\kappa}\}$ such that the points
in $Z_i[t]$ are collectively $\frac{1}{n}$-dependent on them.
Now observe the following:
\begin{enumerate}
\item $\bfz_1$ is $\frac{1}{n}$-dependent on each fault-free process in $\{p_1,p_2,\cdots,p_{f-f_i+1}\}$.
Then, for each $j$, $1\leq j\leq f-f_i+1$,
there exists a convex combination representation of $\bfz_1$
in terms of elements of $\bfv[t-1]$, in which the weight of
process $p_j$ is at least $\frac{1}{n}$.
By ``averaging'' over these $f-f_i+1$ convex combination representations of $\bfz_1$,
we can obtain another convex combination representation of $\bfz_1$ in terms of
the elements of $\bfv[t-1]$ in which
weight of {\em each} process in $\{p_1,p_2,\cdots,p_{f-f_i+1}\}$ is at least $\frac{1}{n(f-f_i+1)}\geq \frac{1}{n^2}$.
\item
When $\kappa\geq 2$,
for $2\leq \nu\leq \kappa$, $\bfz_\nu$ is $\frac{1}{n}$-dependent on fault-free process $p_{f-f_i+\nu}$.
Thus, 
there exists a convex combination representation of $\bfz_\nu$
in terms of elements of $\bfv[t-1]$, in which the weight of
process $p_{f-f_i+\nu}$ is at least $\frac{1}{n}\geq \frac{1}{n^2}$.
\end{enumerate}
Recall that $\bfv_i[t]$ is computed as average of the points in $Z_i[t]$, where
$\{\bfz_1,\bfz_2,\cdots,\bfz_\kappa\}\subseteq Z_i[t]$, and $|Z_i[t]|\leq \nchoosek{n}{(d+1)f+1}$.
Thus, the two observations above imply that there exists a
there exists a convex combination representation of $\bfv_i[t]$
in terms of elements of $\bfv[t-1]$, in which the weight of
each process in $\{p_1,p_2,\cdots,p_{f-f_i+\kappa}\}$ is at least
$\frac{1}{n^2(1+|Z_i[t])}
\geq
\frac{1}{n^2\left(1+\nchoosek{n}{(d+1)f+1}\right)}$.

\begin{eqnarray}
\beta &=&\frac{1}{n^2\left(1+\nchoosek{n}{(d+1)f+1}\right)}
\label{e_beta}
\end{eqnarray}
and define set
$$P_i[t] = \{p_1,p_2,\cdots,p_{f-f_i+\kappa}\}.$$
Note that $|N_i^-\cap (\sv-\sF)|=(d+1)f+\kappa-f_i$. Thus,
\begin{eqnarray}
|P_i[t]| ~=~ f-f_i+\kappa=|N_i^-\cap (\sv-\sF)|-df \label{e_P}
\end{eqnarray}

Recall that we chose $i$ to be any fault-free process in $\sv-\sF$.
Thus, for each fault-free process $i$, such a set $P_i[t]$ exists, where $|P_i[t]|=|N_i^-\cap (\sv-\sF)|-df$.
Therefore, for each fault-free process $i$,
there exists
a convex combination representation of $\bfv_i[t]$
in terms of elements of $\bfv[t-1]$, in which the weight of
each process in
$\{i\}\cup P_i[t]$ is at least $\beta$.
In particular, there exist weights $\alpha_j$'s such that
$\sum_{j\in \{i\}\cup N_i^-}\alpha_j=1$,
 $0\leq \alpha_j\leq 1$ for all 
 $j\in\{i\}\cup N_i^-$, and
\[
\bfv_i[t] ~=~ \sum_{j\in \{i\}\cup N_i^-}\, \alpha_j\,\bfv_j[t-1] 
\mbox{~~and}
\]
\begin{eqnarray}
\label{e_alpha}
\alpha_j~\geq~\beta \mbox{~for~} j\in \{i\}\cup P_i[t].
\end{eqnarray}

Let us now define $i$-th row of matrix $\bfM[t]$ as follows:
\begin{itemize}
\item $\bfM_{ij}[t] = \alpha_j$, for $j \in \{i\}\cup N_i^-$, and
\item $\bfM_{ij}[t] = 0$, otherwise.
\end{itemize}
Due to (\ref{e_P}), the subgraph consisting of only the fault-free processes in $\sv-\sF$,
such that each process $i\in \sv-\sF$
only has incoming links from the processes in $P_i[t]$ is a reduced graph.
Then, defining this subgraph as $H[t]$, the lemma follows from (\ref{e_alpha}).
\end{itemize}
\end{proof}

\section{Proof of Lemma \ref{l_cm}}
\label{a_l_cm}

\paragraph{Lemma \ref{l_cm}:}
{\em
For any $t \geq 1$, there exists a reduced graph $\graphh[t] \in R_{\sF}$ such
 that $\beta {\normalfont\bf\matrixh[t] \leq \matrixm}[t]$,
where $\matrixh[t]$ is the connectivity matrix for $H[t]$.
}
\\

\noindent
\begin{proof}
By Lemma \ref{l_matrix}, there exists
a reduced graph $H[t]$
such that $\matrixm_{ij}[t] ~ \geq ~ \beta$,
if $j=i$ or edge $(j,i)$ is in the reduced graph $H[t]$.

Let $\bfH[t]$ denote the connectivity matrix for reduced graph $H[t]$.
Then
$\bfH_{ij}[t]$ denotes the element in $i$-row and $j$-th column of
$\bfH[t]$.
By definition of the connectivity matrix, we know that $H_{ij}[t]=1$
if $j=i$ or edge $(j,i)$ is in the reduced graph; otherwise,
$H_{ij}[t]=0$. 

The claim in Lemma \ref{l_cm} then follows from the above two observations.
\end{proof}

\end{document}